\documentclass[submission,copyright,creativecommons]{eptcs}

\usepackage{amsmath}
\usepackage{amssymb}
\usepackage{amsthm}
\usepackage{listings}
\usepackage{xspace}
\usepackage{graphicx}
\usepackage{stmaryrd}
\usepackage{tikz}
\usepackage{rotating}

\usepackage{color}
\definecolor{dkblue}{rgb}{0,0.1,0.5}
\definecolor{dkgreen}{rgb}{0,0.4,0}

\definecolor{linkColor}{rgb}{0,0,0.5}

\makeatletter
\@ifclassloaded{beamer}
  {\relax}
  {\newtheorem{theorem}{Theorem}

\newtheorem{lemma}[theorem]{Lemma}

\newtheorem{proposition}{Proposition}

\theoremstyle{definition}
\newtheorem{definition}{Definition}
\newtheorem{example}{Example}

}
\makeatother

\newcommand{\setof}[2]{\ensuremath{\{ #1 \mid #2 \}}}
\newcommand{\set}[1]{\ensuremath{\{ #1 \}}}

\newcommand{\ra}[1][\relax]{\ensuremath \rightarrow_{#1}}

\newcommand{\pra}{\ensuremath \rightharpoonup }

\newcommand{\ol}{\overline}

\newcommand{\Nat}{{\mathbb N}}

\newcommand{\defeq}{\stackrel{\mathrm{def}}{=}}

\newcommand{\dom}[1]{\mathrm{dom}(#1)}
\newcommand{\ran}[1]{\mathrm{ran}(#1)}

\newcommand{\subst}[2]{\{  {}^{#1} / {}_{#2} \}}

\newcommand{\trans}[1]{\xrightarrow{#1}}

\newcommand{\infrule}[2]
           {\parbox{4.5cm}{$$ \frac{#1}{#2}\hspace{.5cm}$$}}

\newcommand{\runa}[1]{\mbox{\textsc{\protect{(#1})}}}

\newcommand{\Names}{\ensuremath{\mathcal{N}}}
\newcommand{\Nil}{\mathbf{0}}

\newcommand{\para}{\mid}

\newcommand{\inpend}[2]{#1(#2)}
\newcommand{\outpend}[2]{{\overline{#1}}\langle #2 \rangle}
\newcommand{\inp}[2]{\inpend{#1}{#2}.}
\newcommand{\outp}[2]{\outpend{#1}{#2}.}

\newcommand{\new}[2][\relax]{(\nu_{#1} #2)}

\newcommand{\fn}[1]{\mathrm{fn}(#1)}
\newcommand{\bn}[1]{\mathrm{bn}(#1)}
\newcommand{\n}[1]{\mathrm{n}(#1)}

\newcommand{\scong}{\ensuremath{\equiv}}

\lstdefinelanguage{spi}
 {morekeywords=[1]{ok,pair,fst,snd,Un,Key},
  morekeywords=[1]{Ok,Pair,Ch,empty,union},
  morekeywords=[2]{out,in,new,event,if,then,else,let,else,begin,end,exercise,decrypt,as,split,match},
  morekeywords=[3]{TermCtor,TermData,TermDtor,Predicate,Process,Private,TopLevel,Main},
  sensitive=true,%
  morecomment=[l]{//},
  literate={->}{$\to\:$}3,
  literate={<=}{$\leq\:$}3,
  morestring=[b]"}

\newcommand{\tabref}[1]{Table \ref{#1}}

\newcommand{\Tenv}{\Gamma}
\newcommand{\lin}{\ensuremath{\mathsf{lin}\,}}
\newcommand{\un}{\ensuremath{\mathsf{un}\,}}
\usepackage{float}

\providecommand{\urlalt}[2]{\href{#1}{#2}}
\providecommand{\doi}[1]{doi:\urlalt{http://dx.doi.org/#1}{#1}}

\providecommand{\lin}{\ensuremath{\mathsf{lin}\;}}
\providecommand{\un}{\ensuremath{\mathsf{un}\;}}

\providecommand{\textfunc}[2]{\ensuremath{\textsf{#1}(#2)}}
\providecommand{\depth}[1]{\textfunc{depth}{#1}}
\providecommand{\nest}[1]{\textfunc{nest}{#1}}
\providecommand{\recs}[1]{\textfunc{recs}{#1}}
\providecommand{\recv}[1]{\textfunc{recv}{#1}}

\providecommand{\sessend}{\mathsf{end}}

\providecommand{\kn}[1]{\ensuremath{\mathrm{kn}(#1)}}

\providecommand{\Rec}{\ensuremath{\mathsf{rec}}}

\providecommand{\chan}[1]{\textfunc{Ch}{#1}}

\providecommand{\Types}{\ensuremath{\mathcal{T}}\xspace}

\newcommand{\outpt}[2]{!\, #1.#2}
\newcommand{\inpt}[2]{?\, #1.#2}

\newcommand{\pol}[2]{\ensuremath{#1}^{#2}}

\newcommand{\Recs}{\ensuremath{\mathcal{R}}\xspace}
\newcommand{\Renv}{\ensuremath{\Delta}\xspace}
\newcommand{\emptyRenv}{\ensuremath{\Renv_{\emptyset}}\xspace}

\renewcommand{\Tenv}{\ensuremath{\Gamma}\xspace}
\newcommand{\Tenvlin}{\ensuremath{\Tenv_{\mathsf{lin}}}\xspace}

\renewcommand{\Names}{\ensuremath{\mathcal{N}}\xspace}
\newcommand{\Polnames}{\ensuremath{\mathcal{N}_{\mathrm{pol}}}\xspace}

\providecommand{\event}{EXPRESS/SOS 2017} 

\title{Using Session Types for Reasoning About Boundedness in the $\pi$-Calculus}
\author{Hans H\"{u}ttel\thanks{E-mail: \texttt{hans@cs.aau.dk}}
\institute{Department of Computer Science, Aalborg University, Denmark}}

\def\titlerunning{Using Session Types for Reasoning About Boundedness in the $\pi$-Calculus}
\def\authorrunning{Hans H\"{u}ttel}
\begin{document}
\maketitle

\begin{abstract}
  The classes of depth-bounded and name-bounded processes are
  fragments of the $\pi$-calculus for which some of the decision
  problems that are undecidable for the full calculus become
  decidable. $P$ is depth-bounded at level $k$ if every reduction
  sequence for $P$ contains successor processes with at most $k$
  active nested restrictions. $P$ is name-bounded at level $k$ if
  every reduction sequence for $P$ contains successor processes with
  at most $k$ active bound names. Membership of these classes of
  processes is undecidable. In this paper we use binary session types
  to decise two type systems that give a sound characterization of
  the properties: If a process is well-typed in our first system, it is
  depth-bounded. If a process is well-typed in our second, more
  restrictive type system, it will also be name-bounded.
\end{abstract}

\section{Introduction}

In the $\pi$-calculus, the notion of name restriction is particularly
important. The study of properties of name binding is a testbed for
studying properties of bindable entities and notions of scoping in
programming languages. In a restriction process $\new{x}P$ the name
$x$ has $P$ as its scope and it is customary to think of $x$ as a new
name, known only to $P$. It is the interplay between restriction and
replication (or recursion) that leads to the $\pi$-calculus being
Turing-powerful. Without either of these two constructs, this is no
longer the case \cite{DBLP:journals/iandc/LanesePSS11}.

With full Turing power comes undecidability of commonly encountered
decision problems such as the
termination problem ``Given process $P$, will $P$ terminate?'' and the
coverability problem ``Given process $P$ and process $Q$, is there a
computation of $P$ that will eventually reach a process that has $Q$
as a subprocess?''.  Several classes of processes have been identified
for which (some of) these problems remain decidable. Examples are the
\emph{finitary processes} without replication or recursion, the
\emph{finite-control processes} \cite{DBLP:journals/iandc/Dam96} in
which every process has a uniform bound on the number of parallel
components in any computation, the \emph{bounded processes}
\cite{DBLP:conf/fossacs/Caires04} for which there are only finitely
many successors of any reduction up to a special notion of structural
congruence with permutation over a finite set of names, and
\emph{processes with unique receiver and bounded input}
\cite{Amadio:2002:DCR:643009.643011}. 

More recently, there has been work in this area that studies
limitations on the use of restriction that will ensure
decidability. The notion of \emph{depth-bounded} processes was
introduced by Meyer in \cite{DBLP:conf/ifipTCS/Meyer08}. A process $P$
is depth-bounded at level $k$ if there is an upper bound $k$, such
that any reduction sequence for $P$ will only lead to successor
processes that have at most $k$ active nested restrictions -- that is,
restrictions not occurring underneath some prefix. Termination and
coverability are both decidable for depth-bounded processes. The class
of depth-bounded processes is expressive and contains a variety of
other decidable subsets of the $\pi$-calculus. Moreover, for any fixed
$k$ it is decidable if a process $P$ is depth-bounded at level $k$;
however, it is undecidable if there exists a $k$ for which $P$ is
depth-bounded \cite{DBLP:conf/ifipTCS/Meyer08}.

In a more recent paper \cite{DBLP:journals/corr/DOsualdoO15},
D'Osualdo and Ong have introduced a type system that gives a sound
characterization of depth-boundedness: If $P$ is well-typed, then $P$
is depth-bounded. The underlying idea of this type system is to
analyze properties of the hierarchy of restrictions within a process.

Another class of $\pi$-calculus processes is that of
\emph{name-bounded} processes, introduced by H\"{u}chting et
al. \cite{Hüchting2013}. A process $P$ is name-bounded at level $k$
if any reduction sequence for $P$ will only lead to successor
processes with at most $k$ active bound names.

The goal of this paper is to use binary session types
\cite{DBLP:conf/esop/HondaVK98} to give sound characterizations of
depth-boundedness, respectively name-boundedness in the
$\pi$-calculus: If a process is well-typed, we know that it is
depth-bounded, respectively name-bounded. The advantages of this
approach are the following: Firstly, unlike the type system proposed
by D'Osualdo and Ong ¨\cite{DBLP:journals/corr/DOsualdoO15} we can
directly keep track of how names are used and where they appear in a
process, since this is central to session type disciplines. The linear
nature of session names ensures that every name of this kind will
always, when used, occur in precisely two parallel
components. Secondly, the session type disciplines are
resource-conscious; we can therefore ensure that new bound names are
only introduced whenever existing bound names can no longer be
used. Both type systems use finite session types to achieve this for
recursive processes. Informally, a new recursive call can only occur
once all sessions involving the bound names of the current recursive
call have been used up. In the proof of the soundness of the system
for characterizing name-boundedness system, we make use of the fact
that it is a more restrictive version of that for depth-boundedness.

The rest of our paper is organized as follows. Section
\ref{sec:picalculus} describes the $\pi$-calculus that we will
consider; section \ref{sec:notions} introduces the notions of
boundedness. Section \ref{sec:depthbounded} presents a type system for
depth-bounded processes, which is analyzed in sections
\ref{sec:subject-reduction} and \ref{sec:soundness}. Section
\ref{sec:namebounded} presents a type system for name-bounded
processes. Section \ref{sec:other} discusses the relationship with
other classes of processes.

\section{A typed $\pi$-calculus with recursion} \label{sec:picalculus}

We follow Meyer \cite{DBLP:conf/ifipTCS/Meyer08} and use a
$\pi$-calculus with \emph{recursion} instead of replication. The
reason behind this choice of syntax is that we would like infinite
behaviours to make use of bound names in a non-trivial manner that
guarantees boundedness properties. In general, the combination of
restriction and replication in $! \new{x} P$ will result in a process
that fails to be name-bounded.

\subsection{Syntax}

We assume the existence of a countably infinite set of names,
\Names, and let $a,b,\ldots$ and $x,y,\ldots$ range over
\Names. Moreover, we assume a countably infinite set of recursion
variables, \Recs, and let $X,Y, \ldots$ range over \Recs.

\subsubsection{Processes} Following \cite{gayhole} we will use a
version of the $\pi$-calculus with \emph{polarized names} in order to
ensure that the endpoints of a channel will not end up in the same
parallel component. We assume polarities ranged over by $p, q
\ldots$. The polarities $+$ and $-$ are dual; we define $\ol{+} = -$
and $\ol{-} = +$. The empty polarity $\varepsilon$ is self-dual and
used for names used as channels that are not session channels and to
tag name occurrences in the binding constructs of input and
restriction. We call the set of polarized names \Polnames.
 
The formation rules of processes are given by
\begin{align*}
P & ::= \inp{\pol{x}{p}}{y}P_1 \mid \outp{\pol{x}{p}}{\pol{y}{q}}P_1 \mid P_1 \para P_2 \mid
\new{x:T}P_1 \mid \mu\, X.P_1 \mid X \mid \Nil \\
p & ::= + \mid - \mid \epsilon
\end{align*}
As usual, $\inp{\pol{x}{p}}{y}P_1$ denotes a process that inputs a
name on channel $x$ and continues as $P_1$; the unpolarized name $y$
is bound in $P_1$. $\outp{\pol{x}{p}}{\pol{y}{q}}P_1$ is a process
that outputs the polarized name $y^q$ on channel $x$ and continues as
$P_1$.  $P_1 \para P_2$ is the parallel execution of $P_1$ and
$P_2$. $\mu X.P_1$ is a recursive process with body $P_1$. We assume
that every such recursive process is \emph{guarded}; every occurrence
of a recursion variable must be found underneath an input or an output
prefix. In $\mu X.P_1$ the $\mu X$ is called a binding occurrence of
$X$. A process $P$ is \emph{recursion-closed} if every recursion
variable $X$ in $P$ has a binding occurrence for some subprocess
$\mu X.P_1$ and if all recursion variables are distinct. We employ a
notion of typed restriction, which we will now explain.

\subsubsection{Typed restrictions} In the restriction $\new{x:T}P_1$
the unpolarized name $x$ is bound in $P_1$ and annotated with type
$T$. Our set of types \Types is a non-recursive version of the binary
session types introduced by Gay and Hole \cite{gayhole} and defined by
the formation rules
\begin{align*} 
T & ::= S \mid \chan{T} \\
S & ::=  (S_1, S_2) \mid \outpt{T}{S} \mid \inpt{T}{S} \mid \sessend
\end{align*}
A type $T$ can be a \emph{linear} endpoint type $S$ or pair of endpoints
$(S_1,S_2)$, or an \emph{unlimited} channel type $\chan{T}$. An endpoint type $S$
of the form $\outpt{T}{S}$ denotes that a channel of this type can
output a name of type $T$; afterwards, the channel will have type
$S$. An endpoint type of the form $\inpt{T}{S}$ denotes that a channel
of this type can input a name of type $T$; afterwards, the channel
will have type $S$. The special endpoint type $\sessend$ is the type
of an endpoint that allows no further communication. If
$T = ( !T_1.S_2, ?T'_1.S'_2)$ we let $T \downarrow = (S_2,S'_2)$; this
denotes the successor of a pair of endpoint types. If $T = \chan{T_1}$,
then $T \downarrow = T$.

We use the type annotation of restrictions to keep track of the
subject name that led to a reduction and of how the types of bound
names evolve. 

The sets of free and bound names of a process, $\fn{P}$ and $\bn{P}$,
are defined as usual. To simplify the presentation, we assume all free
and bound names distinct. We let $P \subst{y}{x}$ denote the
capture-avoiding substitution that replaces all free occurrences of
$x$ in $P$ by $y$. A name $n \in \bn{P}$ is \emph{active} if it does
not appear underneath a prefix.

\subsubsection{Structural congruence}

Structural congruence is the least congruence relation for the process
constructs that is closed under the axioms in \tabref{tab:scong}.

\begin{table}
\begin{center}
  \begin{tabular}{llll}
    \runa{New-1} & $\new{x:T}\new{y:T'} P \scong
                 \new{y:T'}\new{x:T}P$ & \runa{Nil-1} & $P \para \Nil
                                                   \scong P$ \\[2mm]
\runa{New-2} & $\new{x:T}P \para Q \scong \new{x:T}(P \para Q)$ if $x
               \notin \fn{Q}$ \parbox{15mm}{\mbox{}}& \runa{Nil-2} & $\new{x:T}\Nil \scong
                                               \Nil$ \\[2mm]
\runa{Par-1} & $P \para Q \scong Q \para P$ \\[2mm]
\runa{Par-2} & $(P \para Q) \para R \scong P \para (Q \para R)$ 
  \end{tabular}
\end{center}
\caption{Structural congruence: Axioms and rules}
\label{tab:scong}
\end{table}

Following Meyer \cite{DBLP:conf/ifipTCS/Meyer08}, we sometimes
consider processes in \emph{restricted form}. A process is in inner
normal form, if every restriction $\new{x:T}$ only encloses parallel
components that contain $x$.  A process is in outer normal form if
every restriction not underneath a prefix appears at the outermost
level.

\begin{definition}[Normal forms]
Let $P$ be a process.
\begin{itemize}
\item $P$ is in \emph{inner normal form} if for every subprocess
  $\new{x:T}(P_1 \para \cdots \para P_k)$ where none of the $P_i$ are
  parallel compositions of processes, we have $x \in \fn{P_i}$ for all
  $1 \leq i \leq k$.
\item$P$ is in \emph{outer normal form} if
  $P = \new{x_1} \ldots \new{x_k}P_1$ such that $x_i \in \fn{P_1}$ for
  all $1 \leq i \leq k$ and such that all restrictions in $P_1$ 
  appear underneath prefixes.
\end{itemize}
\end{definition}

\begin{proposition}
For every process $P$ we can construct a process $P_1 \scong P$ in
inner normal form and a process $P_2 \scong P$ in outer normal form.
\end{proposition}

\subsection{An annotated reduction semantics}

We define the behaviour of processes by an annotated reduction
semantics that keeps track of when recursive unfoldings are
necessary. Reductions are of the form $P \trans{\alpha} P'$ where
either $\alpha = \set{a}, a \in \Names$ or $\alpha = \set{\Rec,a}$ for
$a \in \Names$. The latter annotation indicates that recursive
unfolding was necessary to obtain the reduction. We define
$\n{\set{a}} = a$ and $\n{\set{\Rec,a}} = a$. The reduction rules are
found in \tabref{tab:red}. Note that in the rule \runa{New-Annot} the
type associated with the bound name $x$ evolves, if $x$ if $x$ is
responsible for the communication and $T$ is a session type.

If $P$ reduces to $P'$ in zero or more reduction steps, we write $P \ra^* P'$.

\begin{table}
\begin{center}
\begin{tabular}{ll}
\runa{Com-Annot} & $\inp{\pol{a}{p}}{x}P_1 \para \outp{\pol{a}{\overline{p}}}{\pol{y}{q}}P_2 \trans{\set{a}} P_1
             \subst{\pol{y}{q}}{x} \para P_2$  \\[3mm]
\runa{Par-Annot} & \infrule{P \trans{\alpha} P'}{P \para Q \trans{\alpha} P' \para Q} \\
\runa{New-Annot} & \infrule{P \trans{\alpha} P'}{\new{x: T}P \trans{\alpha} \new{x:T'}P'}
                   where \begin{tabular}{l} $T = T'$ if $x \notin \alpha$
                           \\ $T' = T \downarrow$
                   if $x \in \alpha$ \end{tabular} \\
\runa{Unfold-Annot} & \infrule{P > Q \quad Q \trans{\alpha} P'}{P
                      \trans{\set{\Rec} \cup \alpha} P'} \\
\runa{Struct-Annot} & \infrule{P \scong Q \quad Q \trans{\alpha} Q' \quad Q' \scong
                P'}{P \trans{\alpha} P'}
\end{tabular}
\end{center}
\caption{Annotated reduction rules}
\label{tab:red}
\end{table}

Recursion is described by an unfolding relation which we define in
\tabref{tab:unfold}.  In the definition, we use the notion of
\emph{unfolding contexts}. An unfolding context $C[\;]$ is an
incomplete process terms whose hole indicates where a prefix that
participates in a reduction step appears as the direct result of
unfolding a recursive process.

\begin{definition}[Unfolding contexts]
  The set of unfolding contexts is given by the formation rules
\[ C ::= [\; ] \! \para \! P \, \mid \, \new{x:T}C \]
\end{definition}

\begin{table}[h]
  \centering
  \begin{tabular}[c]{llp{1cm}ll}
    \runa{Unfold} & $\mu X.P > P[\mu X.P / X]$ & &
    \runa{Context} & \infrule{P > P'}{C[P] > C[P']}
  \end{tabular}
  \caption{The rules for unfolding}
  \label{tab:unfold}
\end{table}

\begin{example}
We can write the process
\[ P \defeq \new{c:T} \mu X. \inp{a}{x}\outp{x}{x}X \para \mu
Y. \new{b:U} \outp{a}{b}\inp{x}{y} Y \]
as 
\[ C_1[\mu X. \inp{a}{x}\outp{x}{x}X] \text{ where } C_1 = [\new{c:T} [ ] \para \mu
Y.  \new{b:U} \outp{a}{b}\inp{x}{y} Y] \]
or 
\[ C_2[\mu Y. \new{b:U}  \outp{a}{b}\inp{x}{y} Y] \text{ where } C_2 = \new{c:T} \mu X. \inp{a}{x}\outp{x}{b}X \para [].\]
\end{example}

\section{Notions of boundedness} \label{sec:notions}

Meyer introduces three notions of boundedness
\cite{DBLP:conf/ifipTCS/Meyer08} for the $\pi$-calculus, and we now
introduce them.

\paragraph{Depth-bounded processes}

A process $P$ is depth-bounded if every configuration reachable from
it can be rewritten so as to have no more than $k$ nested
restrictions. To define this, we first introduce a function $\nest{P}$
that counts the maximal number of active nested restrictions. A
restriction is active if it does not occur underneath a prefix -- this
is similar to \cite{DBLP:journals/corr/DOsualdoO15}.

\begin{definition}The $\textsf{nest}$ function is defined by the clauses
\begin{align*}
  \nest{\Nil}  & = 0  &&
  \nest{X}   = 0 \\
\nest{\new{x:T}P}  & = 1 + \nest{P} && 
\nest{P_1 \para P_2}  = \max(\nest{P_1},\nest{P_2}) \\
\nest{\mu X.P_1}  & = \nest{P_1} &&
\nest{\inp{\pol{x}{p}}{y}P_1}   = \nest{\outp{\pol{x}{p}}{\pol{y}{q}}P_1} = 0
\end{align*}
\end{definition}

The restriction depth of a process is then the minimal nesting depth
up to structural congruence.

\begin{definition}The depth of a process $P$ is given by \[\depth{P} =
    \min \setof{\nest{Q}}{Q \scong P}. \]©
\end{definition}
We define a normalization ordering $\succ$ on processes that removes
bound names not found in a process. It is generated by the axiom
\[  \new{x}P \succ P \quad \text{if } x \not\in \fn{P} \]
and closed under structural congruence. A process $P$ is
\emph{normalized} if it has no superfluous bound names, that is, if
$P \not\succ$; we write $P \leadsto Q$ if $P \succ^* Q$ and $Q$ is
normalized.  \footnote{Note that
  $\new{x:T}P \scong P \quad \text{if } x \not\in \fn{P}$ is a derived
  identity if we include the axiom $\new{x}\Nil \scong \Nil$.}

\begin{definition}[Depth-bounded process]
 A process $P$ is depth-bounded if there is a $k \in \Nat$ such that
for every $P'$ where $P \ra^* P'$ we have that for some $P''$ with
$P'' \scong P'$ we have $\depth{P'} \leq k$.
\end{definition}

\paragraph{Name-boundedness}

A process $P$ is \emph{name-bounded} if there exists a constant
$k \in \Nat$ such that whenever $P \ra^* P'$ and $P' \leadsto P''$, then
$P''$ has at most $k$ restrictions. It is obvious that every
name-bounded process is also depth-bounded.

\begin{example}
The term
\[ P_1 = \mu X. \new{r_1}( \outp{\pol{r_1}{+}}{a}X \para \inp{\pol{r_1}{-}}{x}X) \]
is depth-bounded with $\depth{P_1} = 1$. The term
\[ P_2 = \mu X. \new{r_1}\new{r_2}( \outp{\pol{r_1}{+}}{r_2}X \para \inp{\pol{r_1}{-}}{x}X \para
\outpend{\pol{r_2}{+}}{r_1} \para \inpend{\pol{r_2}{-}}{x}) \]
is depth-bounded with  $\depth{P_2} = 2$. Neither $P_1$ nor $P_2$ is name-bounded.
\end{example}

\paragraph{Width-boundedness}

A third notion of boundedness is that of \emph{width-boundedness}. A
process $P$ is width-bounded if there exists a constant $k \in \Nat$
such that whenever $P \ra^* P'$ we have that every bound name in $P'$
occurs in at most $k$ parallel components. This coincides with the notion of
  \emph{fencing} recently used by Lange et
  al. \cite{Lange:2017:FOG:3009837.3009847} introduced in their
  analysis of Go programs.

\section{Using session types for depth-boundedness}\label{sec:depthbounded}

We now present a session type system that gives a sound
characterization of depth-boundedness. Our account of binary session
types similar to that used by Gay and Hole \cite{gayhole}.

\subsection{Types and type environments}\label{subsec:types}

Our type judgements are of the form $\Tenv , \Renv\vdash P$, where $\Tenv$
contains the type bindings of the free polarized names in $P$. A type
judgment is to be read as stating that $P$ is well-behaved using the
type information found in the type environment $\Tenv$ and the
recursion environment $\Renv$ (explained in Section \ref{sec:rec-env}).

\begin{definition}
  A type environment $\Tenv$ is a partial function
  $\Tenv: \Polnames \pra \Types$ with finite support.
\begin{itemize}
\item $\Tenv$ is \emph{unlimited} if for every
  $x \in \dom{\Tenv}$ we have $\Tenv(x) = \chan{T}$ for some $T$ or
  $\Tenv(x) = \sessend$ 
\item $\Tenv$ is \emph{linear} if for every $x \in \dom{\Tenv}$ we
  have that $\Tenv(x) \neq \chan{T}$ for all $T$. We let \Tenvlin
  denote the largest sub-environment of \Tenv that is linear.
\item  If for every $x \in \dom{\Tenv}$ we have that
  $\Tenv(x) = \sessend$ or $\Tenv(x) = (\sessend,\sessend)$, we say
  that $\Tenv$ is \emph{terminal}.
\end{itemize}
\end{definition}

We define duality of endpoint types in the usual way (note that
duality is not defined for base types).

\begin{definition}[Duality of endpoint types]
  Duality of endpoint types is defined inductively by
  \begin{align*}
    \ol{\outpt{T}{S}}  = ?T.\ol{S} && 
    \ol{\inpt{T}{S}}  = !T.\ol{S}  &&
    \ol{\sessend}  = \sessend
  \end{align*}
\end{definition}

A type $T = (S_1,S_2)$ is \emph{balanced} if $S_1 = \ol{S_2}$. A type
environment $\Tenv$ is balanced if for all $x \in \dom{\Tenv}$ we have
that $\Tenv(x)$ is a balanced type or a base type $B$.

\begin{definition}[Depth of types]
The depth of an endpoint type $S$ is denoted $d(S)$ and is defined
inductively by
  \begin{align*}
    d(\outpt{T}{S})  = 1 + d(S) &&
    d(\inpt{T}{S})  = 1 + d(S) &&
    d(\sessend)  = 0
  \end{align*}
For a type $T = (S_1,S_2)$ we let $d(T) = \max(S_1,S_2)$. For all
other $T$, we define $d(T) = 0$.
\end{definition}
\begin{definition}[Addition of type environments]
  Let $\Tenv_1$ and $\Tenv_2$ be type environments such that
  $\dom{\Tenv_1} \cap \dom{\Tenv_2} = \emptyset$. Then $\Tenv_1 +
  \Tenv_2$ is the type environment $\Tenv$ that satisfies
\[ 
\Tenv(x) = \begin{cases} 
\Tenv_1(x) & \text{if } x \in \dom{\Tenv_1} \setminus \dom{\Tenv_2} \\
\Tenv_2(x) & \text{if } x \in \dom{\Tenv_2} \setminus \dom{\Tenv_1} 
\end{cases}
\]
\end{definition}

\subsection{Recursion and recursion environments}\label{sec:rec-env}

In our type system, recursion variables are typed with type
environments. A \emph{recursion environment} \Renv is a function that
to each recursion variable $X$ assigns a type environment \Tenv. The
idea is that \Tenv will represent the names and associated types needed to
type a process $\mu X.P$.

\begin{definition}\label{def:rec-env}
  A recursion environment \Renv is a partial function $\Renv: \Recs
  \pra (\Polnames \pra \Types)$ with finite support. We let \emptyRenv
  denote the empty recursion environment.
\end{definition}

\begin{definition}  
  Let $\Renv_1$ and $\Renv_2$ be recursion environments where for all
  $X \in \dom{\Renv_1} \cap\dom{\Renv_2}$ we have
  $\Renv_1(X) = \Renv_2(X)$. $\Renv_1 + \Renv_2$ is the recursion
  environment \Renv satisfying
\[ \Renv(X) = \begin{cases}
\Renv_1(X) & \text{if } X \in \dom{\Renv_1} \setminus \dom{\Renv_2} \\
\Renv_2(X) & \text{if } X \in \dom{\Renv_2} \setminus
\dom{\Renv_1} \\
\Renv_1(X) & \text{otherwise} \end{cases}
\]
\end{definition}

\subsection{Type rules}

The set of valid type judgments is defined by the rules in
\tabref{tab:depthbounded}. The type rules differ from the rules from
standard session type systems in their treatment of recursion in two
ways.

The rule \runa{Var} ensures that a recursion variable $X$ can only be
well-typed for $\Tenv$ and $\Renv$ if the type environment $\Tenv_1$ associated
with $X$ mentions all the names in $\Tenv$. Moreover, the rule
requires that the linear part of the type environment must be
\emph{terminal} and that the linear names present when a recursion
variable $X$ is reached include the ones found in the type environment
used to type the process $\mu X.P$. Therefore, when a recursion
variable is reached and a recursive call is made, the restricted names
in the unfolding will be new: the existing sessions have been ``used
up''.

The rule \runa{Chan} ensures that channels that are not session
channels can only be bound within a non-recursive process, as the
recursion environment present must be $\emptyRenv$. Therefore, names
that are not session names cannot accumulate because of recursive
calls and lead to an unbounded restriction depth.

\begin{table}
  \begin{tabular}{lp{6.5cm}ll}
    \runa{In-1} & \infrule{\Tenv, \pol{x}{p} : T_2, y : T_1 , \Renv \vdash P}{\Tenv, \pol{x}{p} :
                \inpt{T_1}{T_2} , \Renv \vdash \inp{\pol{x}{p}}{y}P} &
    \runa{In-2} & \infrule{\Tenv,  \pol{x}{p} : \chan{T_1}, y : T_1 , \Renv
                  \vdash P}{\Tenv, \pol{x}{p} : \chan{T_1} , \Renv
                  \vdash \inp{x}{y}P} \\
 & where $T_1 \neq \sessend$ & & where $T_1 \neq \sessend$ \\
 \runa{Out-1} &
                                                            \infrule{\Tenv,
                                                             \pol{x}{p} : T_2 , \Renv \vdash
                                                            P}{\Tenv,
                                                            \pol{x}{p} :
                                                            \outpt{T_1}{T_2},
                                                            \pol{y}{q} : T_1
                                                            , \Renv \vdash
                                                            \outp{\pol{x}{p}}{\pol{y}{q}}P}
    &
  \runa{Par} & \infrule{\Tenv_1  , \Renv_1 \vdash P_1 \quad \Tenv_2 , \Renv_2 \vdash
               P_2}{\Tenv_1 + \Tenv_2  , \Renv_1 + \Renv_2 \vdash
               P_1 \para P_2}\\
& $T_1 \neq \sessend$ & & \\
\runa{Out-2} & \infrule{\Tenv, x : \chan{T_2}, \pol{y}{q} : T_2 , \Renv \vdash P}{\Tenv, x :
                  \chan{T_2}, \pol{y}{q} : T_2 , \Renv \vdash
               \outp{x}{\pol{y}{q}}P} &
\runa{Session} & \infrule{\Tenv, x^{+} : S, x^{-} : \overline{S} ,
                 \Renv \vdash P}{\Tenv , \Renv \vdash \new{x :
                 (S,\overline{S} )}P}
    \\
& where $T_2$ unlimited & & \\
  \runa{Nil} & $\Tenv , \Renv \vdash \Nil \qquad \Tenv \; \text{unlimited}$
                                             &
 \runa{Var} & $\Tenv, \Renv \vdash X \quad {\begin{array}{l} \Renv(X)
                                               = \Tenv_1 \\ \dom{\Tenv} \subseteq \dom{\Tenv_1} \\
                   \Tenvlin \text{ is terminal} \end{array}}$
    \\[5mm]
\runa{Rec}
                                                          &
                                                            \infrule{\Tenv,
                                                             \Renv, 
                                                            X : \Tenv \vdash
                                                            P}{\Tenv
                                                            , \Renv \vdash \mu
                                                            X.P} &    
\runa{Chan} & \infrule{\Tenv, x : \chan{T} ,
                 \emptyRenv \vdash P}{\Tenv , \emptyRenv \vdash \new{x : \chan{T}}P}     \\                                                  
  \end{tabular}
\caption{Type rules for depth-boundedness}
\label{tab:depthbounded}
\end{table}

The need for private names to be linear inside a recursive process
arises because an unlimited channel can be exploited by a recursive
process to introduce unbounded nesting, as the following example from
\cite{DBLP:journals/corr/DOsualdoO15} illustrates.

\begin{example}
  Consider the following process that cannot be typed; we therefore
  leave out type annotations and polarities in its description. Let 
        \[ P =
  \new{s}\new{n}\new{v}\new{a}(\outpend{s}{a} \para \mu
S. (\inp{s}{x}\new{b} ((\outp{v}{b}\outpend{n}{x} \para
\outpend{s}{b}) \para S))) \]
The process can evolve as follows.
\[ P \ra^{*} \new{s}\new{n}\new{v}\new{a}(P_1 \para \new{b}\new{b'}
((\outp{v}{b}\outpend{n}{a} \para \outp{v}{b'}\outpend{n}{b} \para
\outpend{s}{b'})) \]
where $P_1 = \mu
S. (\inp{s}{x}\new{b}((\outp{v}{b}\outpend{n}{x} \para
\outpend{s}{b}) \para S))$ can introduce further nesting since the
channel $s$ will, when used together with
recursion, be used with an arbitrary number of new names that
cannot be eliminated.
\end{example}

Note that the \runa{Par} rule implies that a process $P$ that can be
typed in a linear environment must be width-bounded with bound $2$,
since every name can then occur in either precisely one or precisely
two parallel components.

Delegation of session names is handled by \runa{Out-1}; session
channels are linear, so the name $\pol{y}{p}$ cannot appear in the
continuation $P$. A special feature of our type system is that
endpoint channels that are no longer usable cannot be delegated. Thus,
in the rules \runa{In-1}, \runa{In-2}, and \runa{Out-1}, the object
type $T_1$ must be different from $\sessend$.

\section{A subject reduction property} \label{sec:subject-reduction}

To show our characterization of depth-boundedness, we state a type
preservation property: For any well-typed process $P$, the type of the
channel that gives rise to a reduction of $P$ will evolve according
to its session type. 

Since this channel may be a restricted channel, we must also describe
how the session types of restricted channels evolve. Every process in
which all bound names are pairwise distinct gives rise to an internal
type environment (Definition \ref{def:internal}) that collects the types of the bound names; this is
an overapproximation of the types of the active names in the
process. This environment is defined as follows.

\begin{definition}\label{def:internal}
  Let $P$ be a process whose bound names are pairwise
  distinct. $\Tenv_P$ denotes the internal type environment of $P$; it
  is defined by the following clauses (where $\pi$ denotes a prefix).
  \begin{align*}
    \Tenv_{P_1 \para P_2}  = \Tenv_{P_1}, \Tenv_{P_2} &&
    \Tenv_{\new{x:T}P}  = x : T, \Tenv_{P} \\
    \Tenv_{\mu X.P}  = \Tenv_{P} &&
    \Tenv_{\pi.P}  = \Tenv_{P} \\
    \Tenv_{X} = \emptyset
  \end{align*}
\end{definition}

The following substitution lemma for variables tells us about the annotated reductions of open
process terms.

\begin{lemma}[Substitution of variables in reductions]
  If $P[\mu X.P / X] \trans{\set{x}} P'$ then $P \trans{\set{x}} P''$, with $P' =
  P''[\mu X.P / X]$.
\end{lemma}
\begin{proof}
Induction in the structure of $P$.
\end{proof}

\begin{lemma}[Substitution of variables in typings of recursion]
Suppose $\Tenv, \Renv \vdash \mu X.P$ and $\Tenv, \Renv \vdash Q$. Then $\Tenv, \Renv
\vdash Q [\mu X.P / X]$.
\end{lemma}
\begin{proof}
  Induction in the structure of $Q$.
  \begin{description}
  \item[$Q = \Nil$:] Trivial.
  \item[$Q = X$:] Immediate, since $Q [\mu X.P /X] = \mu X.P$.
  \item[$Q = Y$ (with $Y \neq X$):] Immediate.
  \item[$Q = Q_1 \para Q_2$:] We must then have concluded $\Tenv, \Renv
    \vdash Q$ using \runa{Par} with premises $\Tenv_1, \Renv \vdash Q_1$ and
    $\Tenv_2, \Renv \vdash Q_2$. By induction hypothesis we then have
\begin{align*}
\Tenv_1, \Renv \vdash Q_1 [\mu X.P / X] \\
\Tenv_2, \Renv \vdash Q_2 [\mu X.P / X] 
\end{align*}
We now use the \runa{Par} rule and get
\[ \Tenv, \Renv \vdash Q_1[\mu X.P /X] \para Q_2[\mu X.P/X] \]
The result now follows by the distributive property of substitution.
\item[$Q = \new{x:T}P_1$:] We must have conclude $\Tenv, \Renv \vdash Q$
  using \runa{Session} with premise $\Tenv, x : S , \Renv \vdash P_1$. By
  induction hypothesis we have that $\Tenv, x : S , \Renv \vdash P_1[\mu
  X.P/X]$. But then by the \runa{Session} rule we get that $\Tenv, \Renv
  \vdash \new{x:T}P_1[\mu X.P/X]$, and we conclude that $\Tenv , \Renv
  \vdash
  Q[\mu X.P/X]$.
\item[$Q = \mu Y. Q_1$:] We must have concluded $\Tenv, \Renv
    \vdash Q$ using \runa{Rec} with premise $\Tenv, \Renv \vdash
    Q_1$. By induction hypothesis we have
\[ \Tenv, \Renv \vdash Q_1 [\mu X.P / X] \]
We can now apply \runa{Rec} to get the desired result.
\item[$Q = \inp{a}{x}Q_1$:] We must have concluded $\Tenv, \Renv
    \vdash Q$ using \runa{In} with premise $\Tenv_1, a : T_2, x : T_1, \Renv
    \vdash Q_1$ and assuming that $\Tenv = \Tenv_1, a : ?T_1.T_2$. By
    applying the induction hypothesis, we get that 
\[ \Tenv_1, a : T_2, x : T_1, \Renv \vdash Q_1 [\mu X.P / X] \]
An application of \runa{In} and the properties of substitution now
gives us the result.
\item[$Q = \outp{a}{x}Q_1$:] Similar to the previous case.
  \end{description}
\end{proof}

We also need a substitution lemma for names.

\begin{lemma}[Substitution of names] \label{lemma:subst-name}
If $\Tenv, x : T, \Renv \vdash P$ and $y \notin n(P)$ then $\Tenv, y
: T, \Renv
\vdash P \subst{y}{x}$.  
\end{lemma}
\begin{proof}
  Induction in the type rules.
\end{proof}

\subsection{A fidelity theorem}

For a binary session type system, subject reduction takes the form of
\emph{fidelity}: the communications in a well-typed process proceed
according to the protocol specified by the channels involved.

\begin{lemma}[Subject congruence and normalization] \label{lemma:sc}
Suppose $\Tenv, \Renv \vdash P$. Then
\begin{itemize}
\item If $P \scong Q$, then also $\Tenv, \Renv\vdash Q$
\item If $P \succ Q$, then also $\Tenv, \Renv \vdash Q$
\end{itemize}
\end{lemma}
\begin{proof}
Induction in the rules defining $\scong$ and $\succ$.
\end{proof}

The fidelity theorem is a type preservation result: It states that the endpoint
types evolve according to the reduction performed. If the name $x$
giving rise to the reduction is free, the annotation of $x$ in the
type environment changes. If $x$ is bound, its annotation in the
restriction $\new{x:T}$ changes to $\new{x:T'}$, where
$T' = T \downarrow$.

\begin{theorem}[Fidelity]\label{thm:fidelity}
Let $\Tenv$ be a balanced type environment and let $P$ be recursion-closed. If $\Tenv, \emptyRenv \vdash P$ and $P
\trans{\alpha} P'$ where $x = \n{\alpha}$ then 
\begin{itemize}
\item if $x \in \fn{P}$ and $\Tenv =
  \Tenv'', x : T$, then $\Tenv', \emptyRenv \vdash P'$ where $\Tenv'$ is balanced
  and $\Tenv' = \Tenv'', x : T\downarrow$
\item if $x\notin \fn{P}$, then $\Tenv, \emptyRenv \vdash P'$ and if $\Tenv_P =
  \Tenv'', x : T$ then $\Tenv_{P'} = \Tenv'', x : T \downarrow$ and
  $\Tenv_{P'}$ is balanced.
\end{itemize}
\end{theorem}
\begin{proof}
Induction in the reduction rules.
\begin{description}
\item[Com-Annot] Here, only the first case is relevant. We know that $P = \inp{\pol{a}{p}}{x}P_1 \para
  \outp{\pol{a}{\ol{p}}}{\pol{y}{q}}P_2$. Since $\Tenv, \emptyRenv  \vdash P$, we must have that $\Tenv =
  \Tenv_1 + \Tenv_2$ where 
\begin{equation} \Tenv_1, \emptyRenv  \vdash
  \inp{\pol{a}{p}}{x}P_1 \label{eq:dom1} \end{equation}
 and 
\begin{equation}\Tenv_2, \emptyRenv  \vdash \outp{\pol{a}{\ol{p}}}{\pol{y}{q}}P_2. \label{eq:dom2} \end{equation}
We must have used \runa{In} to conclude \eqref{eq:dom1}, so we have
$\Tenv_1(\pol{a}{p}) = ?T_1.S$ and, letting $\Tenv_1 =
\Tenv'_1 +  \pol{a}{p} : ?T_1.S$, we have
\begin{equation}\Tenv'_1, \pol{a}{p} : S, x : T_1, \emptyRenv  \vdash
  P_1. \label{eq:dom3} \end{equation} 
Similarly, we must have used \runa{Out} to conclude
\eqref{eq:dom2}. Since $\Tenv$ is balanced, we have
$\Tenv_2(\pol{a}{\ol{p}}) = !T_1.\ol{S}$. By the substitution lemma Lemma
\ref{lemma:subst-name} and \eqref{eq:dom3}, we
have $\Tenv'_1, \pol{a}{p} : S, \pol{y}{q} : T_1, \emptyRenv  \vdash P_1\subst{y}{x}$. Similarly, letting
$\Tenv_2 = \Tenv'_2, \pol{a}{\ol{p}} : !T_1.\ol{S}, \pol{y}{q} : T_1$, we get
$\Tenv'_2, \pol{a}{\ol{p}} : \ol{S} , \emptyRenv \vdash P_2$. An application of \runa{Par} now
gives us that
\[ \Tenv'_1 + \Tenv'_2 + \pol{a}{p}: S, \pol{a}{\ol{p}} : \ol{S}, y :
  T_1, \emptyRenv  \vdash P_1 \subst{y}{x} \para
P_2 \]
The type environment $\Tenv'_1 + \Tenv'_2 + \pol{a}{p}: S, \pol{a}{\ol{p}} : \ol{S}, y :
  T_1$ is balanced, since $\Tenv'_1$ and $\Tenv'_2$ are balanced and
  since $y$ must appear with polarity $\ol{q}$ in one of these
  (because $\Tenv$ is balanced).
\item[Par-Annot] Since $\Tenv , \emptyRenv \vdash P \para Q$, we have that
  $\Tenv_1, \emptyRenv  \vdash P$ where $\Tenv = \Tenv_1 + \Tenv_2$. The result now
  follows easily by an application of the induction hypothesis to the
  reduction $P \trans{a} P'$ and subsequent use of the \runa{Par}
  rule.
\item[New-Annot] There are two cases here: whether $x = a$ or
  $x \neq a$. In both cases, the result follows immediately by the
  induction hypothesis and use of the \runa{Session} rule.
\item[Unfold-Annot] Follows from Lemma \ref{lemma:sc} and a direct 
  application of the induction hypothesis.
\item[Struct-Annot] Follows from Lemma \ref{lemma:sc} and a direct 
  application of the induction hypothesis.
\end{description}
\end{proof}

\section{Soundness of the type system for depth-boundedness} \label{sec:soundness}

In the following we will consider the correctness properties of the
type system for depth boundedness.

\subsection{Properties of unfolding and nesting}

We first establish a collection of properties that hold for arbitrary
processes. Next we show that there are further properties guaranteed
by well-typed processes.

The following lemma describes how reductions occur. Reductions can
happen directly or may need unfoldings.

\begin{lemma}
  Let $P$ be an arbitrary recursion-closed process. 
\begin{enumerate}
\item \label{lemma:caseone} If $P \trans{\set{x}}
  P'$, then there exists an unfolding context $C$ and a process
  $Q$ such that $P \scong C[Q]$ and $P' \scong C[Q']$, and $Q
  \trans{\set{x}} Q'$ is an instance of \runa{Com-Annot}.
\item \label{lemma:casetwo} If $P \trans{\set{\Rec,x}} P'$ then there exists an unfolding context
  $C$ and either $P \scong C[\mu X.Q_1]$ for some $Q_1$ where 
        $Q_1[\mu X.Q_1 / X] \trans{\set{x}} Q'_1$ and $P' \scong C[Q'_1]$
  or $P \scong C[(\mu X. Q_1) \para Q_2] \text{ where } Q_1[\mu X.Q_1
        / X] \para Q_2 \trans{\set{\Rec,x}} Q'_1 \para Q'_2$  is an instance
        of \runa{Com-Annot} and $P' \scong C[Q'_1 \para Q'_2]$.
\end{enumerate}
\end{lemma}
\begin{proof}
By induction in the annotated reduction rules. The proof of Case
\ref{lemma:casetwo} uses Case \ref{lemma:caseone}.
\end{proof}

\subsection{Nesting properties of well-typed processes}

We now restrict our attention to well-typed processes. The only
potential source of unbounded restriction depth is the presence of
recursion, and we now show how our type system controls the
introduction of new bound names in the presence of recursion.

The first lemma tells us that bound names introduced by an
unfolding do not interfere with names in its surrounding process that
represent terminated channels.

\begin{lemma}
If $\Tenv, \Renv \vdash \new{c: (\sessend,\sessend)}P$ then $c \not \in \fn{P}$.
\end{lemma}

The following lemma tells us that names that appear in an unfolding context
will not reappear free in the result of unfolding a recursive process.

\begin{definition}[Known bound names]
  The set of known bound names in an unfolding context is defined by
  \begin{align*}
    \kn{[] \para P}  = \emptyset && 
    \kn{\new{x:T}C} = \set{x} \cup \kn{C}
  \end{align*}
\end{definition}

\begin{lemma}\label{lemma:newnames}
  Suppose we have $\Tenv, \Renv \vdash C[X]$ where $C[X]$ is
  recursion-closed and $X$ occurs in $\mu X.P$. Then we
  also have $\Tenv, \Renv \vdash C[\mu X.P]$ and
  $\kn{C} \cap \fn{\mu X.P} = \emptyset$.
\end{lemma}

\begin{theorem}
Let $P$ be recursion-closed. Suppose $\Tenv, \emptyRenv \vdash P$. Then $P$ is depth-bounded.
\end{theorem}

\begin{proof}[Outline]
  The session types provide a bound on the nesting depth of a
  well-typed process. Suppose $\Tenv, \emptyRenv \vdash P$. Let
  $d(\Tenv,P)$ denote the sum of the depths of the session types in
  $\Tenv$ and in $\Tenv_P$, i.e.
\[ d(\Tenv,P) = \sum_{x : T \in \Tenv \text{ or } x : T \in \Tenv_P} d(T) \]
In a process $P$ with $k$ bound names, we know from Theorem
\ref{thm:fidelity} that there can be at most
$(d(\Tenv + \Tenv_P)/2) - k$ reduction steps before an unfolding has
to take place, since every reduction step will decrease the depth of
one of the session types in $\ran{\Tenv} \cup \ran{\Tenv_P}$. Whenever
unfoldings occur, the bound names in the unfolding are distinct from
those already known and will all be names of session
channels. Moreover, when the unfolding is reached, the channel used in
the reduction will no longer be available. As a consequence we see
that the nesting depth will therefore not increase.
\end{proof}

\section{A type system for name-boundedness}\label{sec:namebounded}

We now show to modify our previous type system such that every
well-typed process will be name-bounded. The challenge is again one of
controlling recursion. As before, the crucial observation is that if
private channels are linear, then all the channels that have been used
when a recursion unfolding takes place, can then be discarded.

In the case of name-boundedness, extra care must be taken, since
recursion may now accumulate an unbounded number of finite components
that each contain pairwise distinct bound names.

\begin{example}
  The untyped process
  \[ P_2 = \mu X. \new{r_1}\new{r_2}( \outp{r_1}{a}X \para
  \inp{r_1}{x}X \para \outpend{r_2}{a} \para \inpend{r_2}{x}) \]
shows two problems that must be dealt with. Firstly, unfolding a
recursion may introduce more parallel recursive components that each
have their own bound names. In this case, every communication on $r_1$
will introduce two new parallel copies of the recursive
process. Secondly, unfolding may introduce finite (non-recursive)
components which contain bound names that persist -- in this case, we
get new copies of $\new{r_2} (\outpend{r_2}{a} \para \inpend{r_2}{x})$
for every unfolding.
\end{example}

The type language is
\begin{align*}
S_\lin ::= \inpt{T_\lin}{S_\lin} \mid \outpt{T_\lin}{S_\lin} \mid \sessend  &&
S_\un ::= \chan{S_\un}  \\
T_\lin ::=  (S_\lin,S_\un) \mid (S_\lin,S_\lin) && 
T ::= T_\lin \mid S_\un 
\end{align*}
Note that names of unlimited type $S_\un$ can only be used to delegate
channels of unlimited type. 


The type rules are as in the original type system, but we now modify the
notions of addition for type environments and for recursion
environments. We add pairs $(\Tenv_1,\Renv_1)$ and
$(\Tenv_2,\Renv_2)$ as follows.

\begin{definition}
  Let $\Tenv_1,\Tenv_2$ be type environments and let $\Renv_1,\Renv_2$
  be recursion environments where at least one of $\Renv_1, \Renv_2$
  is $\emptyRenv$. We define
  $(\Tenv_1,\Renv_1) + (\Tenv_2,\Renv_2) = (\Tenv_1 + \Tenv_2,\Renv_1
  + \Renv_2)$ where $\Tenv_1$ is unlimited if $\Renv_1 = \emptyRenv$
  and $\Tenv_2$ is linear if $\Renv_2 \neq \emptyRenv$.
\end{definition}

The intention is that an empty recursion environment must now go
together with an unlimited type environment. In other words:
Non-recursive subprocesses can only contain unlimited names.

We say that a type environment $\Tenv$ is \emph{limited} if for every
$x \in \dom{\Tenv}$ we have that $\Tenv(x) =
(T_\lin,\ol{T_\lin})$ for some $T_\lin$. That is, the environment is balanced, and no
name has an unlimited type.

A type environment $\Tenv$ is \emph{skew} if $\Tenv = \Tenv_1 +
\Tenv_2$ with $\dom{\Tenv_1} \cap \dom{\Tenv_2} = \emptyset$,
$\Tenv_1$ is linear and for all $x \in \dom{\Tenv_2}$ we have that
$\Tenv(x) = (T_\lin,T_\un)$ for some $T_\lin, T_\un$.

\subsection{Fidelity}

As in the case of the previous type system, we need a fidelity result.

\begin{theorem}[Fidelity]\label{thm:fidelity-2}
Let $\Tenv$ be a type environment. If $\Tenv , \emptyRenv \vdash P$
and $P \trans{x} P'$ then 

\begin{itemize}
\item if $x \in \fn{P}$ and $\Tenv =
  \Tenv'', x : T$, then $\Tenv', \emptyRenv \vdash P'$ where $\Tenv'$ is balanced
  and $\Tenv' = \Tenv'', x : T\downarrow$
\item if $x\notin \fn{P}$, then $\Tenv, \emptyRenv \vdash P'$ and if $\Tenv_P =
  \Tenv'', x : T$ then $\Tenv_{P'} = \Tenv'', x : T \downarrow$ and
  $\Tenv_{P'}$ is balanced.
\end{itemize}
\end{theorem}

Since the new type system specialized the previous one, this result is
easily established.

\subsection{Soundness for name-boundedness}

We will show that if a process is well-typed in a limited environment,
then it is name-bounded. 


To show that a well-typed process $P$ is name-bounded, we will show that

\begin{itemize}
\item For some $k$, whenever $P \ra^* P'$, then $P'$ has at most $k$
  recursion instances in $P'$
\item For some $m$, whenever $P \ra^* P'$, every recursive subprocess
  of $P'$ contains at most $m$ distinct bound names
\item There are only free names in the non-recursive part of $P$
\end{itemize}

Since every well-typed process is known to be depth-bounded, the
result will then follow.

Our first lemma gives a characterization of well-typed recursive
processes: They can contain at most one instance of each recursion
variable. 

\begin{lemma} \label{lemma:one-rec-only}
Let $\mu X.P$ be a process for which all binding occurrences of
recursion variables are distinct. If $\Tenv, \Renv \vdash \mu X.P$,
there is at most one occurrence of $X$ in $P$.
\end{lemma}
\begin{proof}
  Suppose to the contrary that there is more than one occurrence of
  $X$ in $P$. We then have that $\mu X. P = \mu X. \new{\vec{n}}(C_1[X] \para C_2[X] \para P')$
where $n$ is a set of names (possibly empty), and $C_1$ and $C_2$ are
process contexts.

The derivation of the type judgement $\Tenv, \Renv \vdash \mu
X. \new{\vec{n}}(C_1[X] \para C_2[X] \para P')$ must have used the
\runa{Rec} type rule in its final step, having premise $\Tenv, \Renv, X : \Tenv \vdash \new{\vec{n}}(C_1[X] \para
C_2[X] \para P'$. But the derivation of this judgement must have used the \runa{Session} rule a number
of times, preceded by an application of \runa{Par} with premises
$\Tenv_1, \Renv, X : \Tenv \vdash C_1[X]$ and $\Tenv_2 , \emptyRenv \vdash C_2[X]$
where $\Tenv_2$ is unlimited. However, there can be no derivation of
the latter, since this would require the rule
\runa{Var} in which it is assumed that the type environment is linear.

We therefore conclude that our initial assumption was wrong; there can
be at most one occurrence of $X$ in $P$.
\end{proof}

This lemma tells us that there can be no finite, non-recursive
subprocesses of a recursive process with their own bound names; any
bound name found in a non-recursive subprocess will also appear in the
recursive part of the process.

\begin{lemma}
  If $\Tenv , \Renv \vdash \mu X. (C[X] \para P)$ where
  $\mu X. (C[X] \para P)$ is in inner normal form and $C[X]$ is a
  process context, then for every $n \in \bn{P}$ we have that
  $n \in \bn{C[X]}$.
\end{lemma}
\begin{proof}
  Consider a name $n \in \bn{P}$. Suppose $n \notin \bn{C[X]}$. Since
  $\mu X. (C[X] \para P)$ is in inner normal form, we would then have
  a subprocess $\new{n:T}P'$ of $P$ that would be typed using the
  \runa{Session} rule. But for this rule to be applicable, a recursion
  variable must be present in the type environment. This cannot be the
  case, as $P$ is non-recursive.
\end{proof}

We now show that the number of recursive subprocesses that will appear
in any reduction sequence for a well-typed process is bounded. Let
$\recs{P}$ denote the number of simultaneous recursion instances in
$P$ and let $\recv{P}$ denote the multiset of recursion variable
occurrences in $P$. 

Together, the following two lemmas give an upper bound on the number
of recursion instances in any reduction sequence of a well-typed
process.

\begin{lemma}
Suppose $\Tenv , \Renv \vdash P$ and $P \trans{\alpha} P'$ was proved without using
instances of \runa{Unfold-Annot}. Then $\recs{P} \geq \recs{P'}$.
\end{lemma}

\begin{lemma}\label{lemma:recs}
Suppose $\Tenv, \Renv \vdash P$ where $\dom{\Renv} \cap \recv{P} =
\emptyset$ and $P > P_1$. Then $\recs{P} \geq \recs{P_1}$.\end{lemma}



The following normal form theorem is crucial.

\begin{theorem}
  If $\Tenv , \Renv \vdash P$, then there exists a $k \geq 0$ such
  that whenever $P \ra^* P'$, we have $P \scong P_1 \para P_2$ where
  $\recs{P_1} \leq k$, $\recs{P_2} = 0$ and $P_2$ contains no
  restrictions.
\end{theorem}
\begin{proof} We show that for all $n \geq 0$, if $P \ra^n P'$, then we have $P
\scong P_1 \para P_2$ where $\recs{P_1} \leq k$, $\recs{P_2} = 0$ and
$P_2$ contains no restrictions. The proof of this proceeds by induction in $n$.

\begin{description}
\item[$n=0$:]  Here we let $k = \recs{P}$ and proceed by induction in
  the type derivation of $\Tenv , \Renv \vdash P$. We consider each rule in
  turn.
  \begin{description}
  \item[\runa{In-1}, \runa{In-2}, \runa{Out-1} and \runa{Out-2}:] None
    of these rules could have been used, since $P$ would
    then have no reductions.
  \item[\runa{Par}:] Here we can use the commutativity and
    associativity axioms for structural congruence to rewrite $P$ in
    the desired form.
  \item[\runa{Var}:] Cannot apply, since we assume that $\Tenv , \Renv \vdash P$.
  \item[\runa{Rec}, \runa{Nil}, \runa{Session}:] These are immediate.
  \end{description}
\item[Assume for $n$, prove for $n+1$:] This is a straightforward
  induction in the type rules.
\end{description}

\end{proof}

\begin{theorem}
If $\Tenv , \Renv \vdash P$, then $P$ is name-bounded.
\end{theorem}
\begin{proof}
  There is a $k \geq 0$ such that if $\Tenv , \Renv \vdash P$, whenever
  $P \ra^* P'$, there are at most $k$ recursive subprocesses of
  $P'$. Since the new type system is a subsystem of the type system
  for depth-boundedness, there exists a $d$ such that the recursion
  depth of $P'$ is at most $d$ for any such $P'$.

  Every bound name in a non-recursive subterm of a recursive
  subprocess occurs in the recursion part as well. Now consider an
  outer normal form $P''$ of $P'$. We have
  $P'' = \new{x_1} \ldots \new{x_d} P^{(3)}$ for some $P^{(3)}$ that
  does not contain restrictions at the outermost level. Moreover, for
  some $k' \leq k$ we have
  $P^{(3)} \scong P^{(3)}_1 \para \cdots P^{(3)}_{k'} \para
  P^{(3)}_{k'+1}$
  where $P^{(3)}_1, \ldots, P^{(3)}_{k'}$ contain recursion instances
  and $P^{(3)}_{k'+1}$ is a process not containing recursion
  instances. We know that for some $d$ there are at most $d \cdot k$
  bound names in $P'$.
\end{proof}

\section{The relation to other classes of processes} \label{sec:other}

Because of the use of binary session types, typable process in our
systems will be width-bounded with name width $2$. On the other hand,
both type systems allow us to type processes that are not
finitary. The classes of typable processes differ from those already
studied. The process
$P_1 \defeq (\mu X. \new{a} \inp{a}{x}X \para \outpend{a}{b} \para
\outpend{b}{c})$ is not a finite-control process, since the reduction
sequence
$P \ra^{k} P_1 \para \outpend{b}{c} \para \cdots \para
  \outpend{b}{b}$ that results in $k-1$ parallel components, each being a simple
  output, shows that the number of
parallel components along a computation can be unbounded for a
well-typed process. This means that $P_1$
is neither a finite-control process \cite{DBLP:journals/iandc/Dam96}
nor a bounded process in the sense of
\cite{DBLP:conf/fossacs/Caires04}. On the other hand, $P_1$ is
depth-bounded, and in fact also width-bounded as every bound name
occurs in precisely two parallel components. Moreover, the typable processes are
incomparable with the processes studied in
\cite{Amadio:2002:DCR:643009.643011} since these do not allow for
delegation of input capabilities.

\section{Conclusions and ideas for further work}

In this paper we have presented two session type systems for a
$\pi$-calculus with recursion. One guarantees depth-boundedness, and
the other system, which is a subsystem of it, guarantees
name-boundedness. Both systems assume that names are always used in
finite-length sessions before a recursive call is initiated.

In the paper by D'Osualdo and Ong
\cite{DBLP:journals/corr/DOsualdoO15} a type inference algorithm is
proposed that makes it possible to provide a safe bound on the
restriction depth for depth-bounded processes. A further topic of
investigation is to adapt the type inference algorithm proposed in
\cite{DBLP:conf/wsfm/GraversenHHBPW15} to the setting of the type
systems of the present paper. We conjecture that this is
straightforward. The type systems presented in this paper are simpler
than many other session type systems, in that they do not involve
recursive types; the sole difference is that of the presence of
recursion instead of replication in the $\pi$-calculus.

In both systems, the number of parallel components in a well-typed
system can be unbounded, and well-typed processes need not be
finite-control. Conversely, finite-control processes need not be
well-typed in the present systems, since finite-control processes are
not necessarily width-bounded with width $2$.

Another important question to be answered is that of the exact
relationshop between our type system for depth-boundedness and the
type system due to D'Osualdo and Ong
\cite{DBLP:journals/corr/DOsualdoO15}.

\end{document}